\documentclass[11pt]{amsart}

\usepackage[margin=1in]{geometry}
\usepackage{amsmath,amssymb,amsthm,mathtools,mathrsfs}
\usepackage{enumitem}
\usepackage{hyperref}
\hypersetup{hidelinks, pdftitle={The Critical Renyi Order for Sharp Minimax Channel Design}, pdfauthor={Alex Shvets}}
\usepackage{microtype}
\usepackage{booktabs}

\usepackage{aliascnt}

\newtheorem{theorem}{Theorem}[section]

\newaliascnt{proposition}{theorem}
\newtheorem{proposition}[proposition]{Proposition}
\aliascntresetthe{proposition}

\newaliascnt{lemma}{theorem}
\newtheorem{lemma}[lemma]{Lemma}
\aliascntresetthe{lemma}

\newaliascnt{corollary}{theorem}
\newtheorem{corollary}[corollary]{Corollary}
\aliascntresetthe{corollary}

\theoremstyle{definition}

\newaliascnt{definition}{theorem}

\aliascntresetthe{definition}

\newaliascnt{convention}{theorem}

\aliascntresetthe{convention}

\newaliascnt{remark}{theorem}
\newtheorem{remark}[remark]{Remark}
\aliascntresetthe{remark}

\usepackage[capitalise,nameinlink,noabbrev]{cleveref}
\crefname{theorem}{Theorem}{Theorems}
\crefname{proposition}{Proposition}{Propositions}
\crefname{lemma}{Lemma}{Lemmas}
\crefname{corollary}{Corollary}{Corollaries}
\crefname{definition}{Definition}{Definitions}
\crefname{convention}{Convention}{Conventions}
\crefname{remark}{Remark}{Remarks}

\DeclareMathOperator{\tr}{tr}
\DeclareMathOperator{\Var}{Var}

\DeclareMathOperator{\argmax}{argmax}
\DeclareMathOperator{\Cov}{Cov}
\newcommand{\one}{\mathbf 1}
\newcommand{\R}{\mathbb R}
\newcommand{\E}{\mathbb E}
\newcommand{\PP}{\mathbb P}
\newcommand{\cP}{\mathcal P}
\newcommand{\cX}{\mathcal X}
\newcommand{\cY}{\mathcal Y}

\newcommand{\fM}{\mathfrak M}
\newcommand{\cK}{\mathcal K}

\newcommand{\cS}{\mathcal S}

\title[The Critical R\'enyi Order]{The Critical R\'enyi Order for Sharp Minimax Channel Design}
\author{Alex Shvets}
\address{Independent Researcher, Haifa, Israel}
\email{alex@shvets.io}
\date{24 October 2026}
\begin{document}
\begin{abstract}
We study categorical channel design under a bound on the maximum directed
pairwise R\'enyi divergence. We determine exactly which orders make
randomized response optimal in every dimension and at every finite budget:
the orders strictly above one and at most two. In that range we identify
the full information-trace frontier and its unique optimal anchored law.
The channel is visibly thinned randomized response below a uniquely
defined threshold and ordinary randomized response above it. A
fixed-second-moment inequality and a second-order variational argument
reduce arbitrary channel columns to singleton templates. Above order two,
an explicit subset-selection channel beats every mixture of singleton
templates in a finite dimension. The trace frontier also gives the exact
finite-sample minimax risk among all unbiased real-valued vector
estimators, and the asymptotic minimax risk among all estimators for every
budget sequence with diverging total information. Appendices compare
pure local privacy and distinguish homogeneous shuffle contamination
from arbitrary neighboring datasets.
\end{abstract}
\maketitle

\section{Introduction}

A channel $W:[d]\to\cP(\cY)$ randomizes each categorical observation before
it is used for estimation. We fix $d\ge3$, use a common channel for all
users, and constrain the distinguishability of its input rows by
\[
\max_{i\ne j}D_\alpha\bigl(W(\cdot\mid j)\|W(\cdot\mid i)\bigr)\le R,
\qquad \alpha>1.
\]
Equivalently, the maximum directed power moment minus one is at most
$C=e^{(\alpha-1)R}-1$. We work in this linear budget $C$. The
change of variables is increasing, but mixtures and Jensen inequalities
are linear in $C$ and are not asserted in $R$.

Write $\bar W=d^{-1}\sum_iW(\cdot\mid i)$ and
$a_i=dW_i/d\bar W$. Thus $\sum_i a_i=d$ and $\E_{\bar W}a_i=1$.
The information trace at the uniform input is
$\E_{\bar W}S(a)$, where $S(a)=\sum_i a_i^2-d$.
A singleton template has one coordinate $h$ and $d-1$ equal coordinates
$l<h$. Its uniform permutation orbit gives generalized randomized
response (GRR). Mixing that orbit with the constant template gives
augmented randomized response (ARR): an input-independent null output
with a visible activation indicator. More generally, the
\emph{singleton family} consists of all anchored laws supported on
singleton templates and the constant template, including arbitrary
mixtures of their parameters.

\begin{theorem}[Critical order; Theorem A]\label{thm:critical}
Fix $\alpha>1$. Randomized response, with visible null thinning below a
threshold and without thinning above it, maximizes the information trace
for every $d\ge3$ and every finite linear budget $C\ge0$ if and only if
$1<\alpha\le2$. In this range the optimal anchored law is unique and also
minimizes worst-case canonical squared error in both sampling models
defined below.

For each $\alpha>2$ there are explicit finite $d$ and $C$ and a
subset-selection channel whose trace strictly exceeds that of every
feasible law in the entire singleton family. This negative conclusion
does not describe the full class of optimizers above order two.
\end{theorem}

\begin{proof}
The positive statement, including the complete frontier and equality
cases, is \cref{thm:global-frontier}; the explicit negative construction
is \cref{thm:above-two}. Their proofs establish \cref{thm:critical}.
The trace statement and the canonical design statement are equivalent
here because permutation averaging makes the information matrix scalar
and cannot increase worst-case canonical risk
(\cref{thm:symmetrization}). No corresponding exact finite-sample
claim for arbitrary biased estimators is included in this theorem.
\end{proof}

\subsection{The frontier and its statistical consequences}
For $1<\alpha\le2$ let $F_{\alpha,d}(C)$ denote the maximum trace.
We determine it by a one-variable cost curve. It has a linear lower
branch attained by ARR and a strictly concave upper branch attained by
GRR, with a unique tangency threshold. The optimizing anchored law is
unique among all channels, including channels without permutation
symmetry; the equivalence is up to common, input-independent refinements
of the output.

The same frontier gives two statistical statements with different
domains. For every finite $n\ge1$ and $0<C<\infty$, the exact i.i.d.
minimax risk over all globally unbiased $\R^d$-valued estimators is
\[
\frac{(d-1)^2}{nF_{\alpha,d}(C)}.
\]
For all estimators, the same expression is asymptotically sharp whenever
$nF_{\alpha,d}(C_n)\to\infty$, with $d$ and $\alpha$ fixed.
The proofs use, respectively, a score identity at the uniform input and
a channel-uniform van Trees inequality. The finite-sample canonical
fixed-composition optimum is also explicit. Global unbiasedness is
essential to the finite-sample result; a simplex-valued globally
unbiased estimator does not exist at a finite budget.

\subsection{Relation to earlier work}
R\'enyi divergence and local R\'enyi privacy are established notions.
Mironov \cite{mironov} develops R\'enyi differential privacy.
The maximum-pair local constraint used here appears in
Duchi and Ruan \cite[Definition~2]{duchi-ruan},
Gilani, Kurri, Kosut, and Sankar \cite[Definition~7]{gilani-lrdp},
and Wang and Xu \cite[Definitions~3--4 in the supplementary material]{wang-xu-regression}.
Our subject is the exact categorical channel-design frontier and its
critical order, rather than the definition of local R\'enyi privacy.
Order-two local privacy also figures in the local minimax theory of
\cite{duchi-ruan}; those bounds address a different level of exactness
and do not supply the frontier proved here.

The normalized-column representation is the posterior, or balanced
Blackwell-measure, representation of an experiment
\cite{blackwell,nasser}. We give its reconstruction proof.
For pure local privacy, Nam, Park, and Lee
\cite[Theorems~3.7 and 3.10]{nam-park-lee-geometry} develop general
Blackwell-geometric and group-symmetry reductions for optimal decision
risks. Their constraint is a pointwise likelihood-ratio cap.
Here we prove the required symmetrization directly for the recalibrated
canonical estimator under a finite-order moment budget.

The fixed-second-moment singleton comparison is a consequence of the
entropy-extremum theorem of Berry and Sanders
\cite[Theorem~1]{berry-sanders}; a positive affine change of variables
also covers its negative-order factor. Related positive-order norm
bounds are given by Sakai and Iwata \cite[Theorem~1]{sakai-iwata}.
We include a self-contained second-order proof and explain the precise
connection after \cref{lem:fixed-B}. The channel-design analysis then
determines the entire concave frontier and its equality cases, and the
higher-order construction determines the boundary of universal
singleton optimality.

The author's earlier paper \cite{shvets-growing} contains the
order-two aggressive randomized-response parameter $\sqrt{d-1}$,
visible thinning, and mechanism-family trace calculations.
The fixed-composition privacy asymptotics in \cite{shvets-fixed} concern
neighboring shuffled experiments; the present statistical objective is
estimation risk. The Fisher-information-trace route to private minimax
lower bounds through van Trees is used by Barnes, Chen, and
\"Ozg\"ur \cite[equations~(2)--(3)]{barnes-chen-ozgur}.
Our exact finite-sample unbiased and asymptotic all-estimator conclusions
use classical score and van Trees arguments with the frontier proved
here; these statistical tools are established methods.

Under a pure local differential privacy cap, staircase mechanisms were
studied by Kairouz, Oh, and Viswanath \cite{kov-extremal}, and
subset-selection estimators by Ye and Barg \cite{ye-barg-subset}.
Ye and Barg \cite{ye-barg-sharp} establish sharp asymptotic minimax
constants for distribution estimation at a fixed pure cap.
Park, Nam, and Lee \cite[arXiv version, Theorems~3.3 and 3.9]{park-nam-lee}
give the exact finite-sample i.i.d. formula and optimality among all
unbiased pure-local-DP schemes. We credit their block-replication
argument when discussing the relation between asymptotic and exact
unbiased bounds. Appendix~\ref{raw:appendix} compares these known
pure-cap formulas with the R\'enyi frontier; subset selection under a
pure ratio cap is a different optimization problem.

\subsection{Sampling and privacy scope}
For i.i.d. inputs, the messages are i.i.d. and their unordered sample is
sufficient for the distribution parameter. Shuffling therefore changes
none of the estimation risks considered here. For deterministic input
compositions the canonical estimator also uses only the sum of messages.
These facts do not give a worst-case privacy theorem for arbitrary
neighboring shuffled datasets.

Appendix~\ref{app:contamination} treats a separate, homogeneous
one-coordinate contamination experiment. Its forward chi-squared
divergence divides exactly by the number of users; its reverse divergence
need not do so. A binary randomized-response envelope holds in that
restricted experiment by classical binary simulation
\cite{kov-composition,murtagh-vadhan}. An exact heterogeneous
three-user counterexample rules out its extension to arbitrary neighbors.
The ancillary file \texttt{anc/renyi\_verification.py} reproduces the
numerical table and the finite algebraic and contamination checks; all
proofs are contained in the manuscript.

\section{Channel coordinates and canonical risk}\label{sec:coordinates}

\subsection{The channel as a mean-zero law}
Write $\Delta_d=\{\theta\in\R_+^d:\one^\top\theta=1\}$ and
$T_d=\{v\in\R^d:\one^\top v=0\}$. Fix a matrix
$H\in\R^{d\times(d-1)}$ whose columns are an orthonormal basis of $T_d$,
and put $\gamma_i=H^\top e_i$. Then
\begin{equation}\label{eq:simplex}
\sum_i\gamma_i=0,\qquad
\gamma_i^\top\gamma_j=\delta_{ij}-1/d,\qquad
\sum_i\gamma_i\gamma_i^\top=I_{d-1}.
\end{equation}
Define
\[
\cX_d=\{x\in\R^{d-1}:1+\gamma_i^\top x\ge0\ (i\in[d])\},
\qquad \cK_d=\{a\in\R_+^d:\one^\top a=d\}.
\]
The maps $a=\one+Hx$ and $x=H^\top(a-\one)$ are inverse affine
bijections between these compact simplices. In particular,
\begin{equation}\label{eq:norm-bound}
\|x\|^2=\sum_i a_i^2-d,\qquad
\sup_{x\in\cX_d}\|x\|=R_d:=\sqrt{d(d-1)}.
\end{equation}
Indeed, $a-\one\in T_d$, on which $H^\top$ is an isometry; the
maximum occurs at a vertex $de_i$ of $\cK_d$.

An anchored law is a Borel probability measure $\rho$ on $\cX_d$ with
$\int x\,\rho(dx)=0$. Let $\fM_d$ be the set of such laws. It defines
the channel
\begin{equation}\label{eq:channel}
W_\rho(dx\mid i)=(1+\gamma_i^\top x)\rho(dx).
\end{equation}
We also write $W_i=W_\rho(\cdot\mid i)$.
Nonnegativity follows from the support, and each row has mass one by the
mean-zero condition.

\begin{proposition}[Posterior coordinates]\label{prop:representation}
Every channel with finite or standard Borel output space has a sufficient representation
\eqref{eq:channel} with $\rho\in\fM_d$. For a finite output space $\rho$ has finite support.
The law $\rho$ is unique after merging outputs with identical normalized
columns. Conversely, every finitely supported anchored law defines a
finite-output channel. For an input distribution $\theta$ the mixture law
is
\begin{equation}\label{eq:mixture}
q_\theta(dx)=(1+h_\theta^\top x)\rho(dx),\qquad
h_\theta:=H^\top(\theta-\one/d)=\sum_i\theta_i\gamma_i.
\end{equation}
\end{proposition}

\begin{proof}
For a finite output space let $\bar W(y)=d^{-1}\sum_iW(y\mid i)$; discard
outputs with $\bar W(y)=0$. Set $a_i(y)=W(y\mid i)/\bar W(y)$,
$x(y)=H^\top(a(y)-\one)$, and let $\rho$ be the law of $x(y)$ under
$\bar W$. Since $\sum_ia_i(y)=d$ and
$\sum_y\bar W(y)a_i(y)=1$, its support is contained in $\cX_d$ and its mean is zero.
The induced row law is \eqref{eq:channel}. If several outputs have the
same $x$, their conditional probabilities given that $x$ are
$\bar W(y)/\sum_{z:x(z)=x}\bar W(z)$, independently of the input row.
Thus the original output can be reconstructed by a common randomization.
Conversely, the row average of \eqref{eq:channel} is exactly $\rho$;
normalizing its columns recovers $a_i(x)=1+\gamma_i^\top x$, hence $x$
itself. For a standard Borel output space, use the densities
$a_i=dW_i/d\bar W$ with $\bar W=d^{-1}\sum_iW_i$ and push
$\bar W$ forward by the same map $x=H^\top(a-\one)$. A regular
conditional law of the original output given $x$ under $\bar W$
reconstructs every row: its density $a_i$ is a function of $x$, so this
conditional law is independent of the input row. The same reconstruction
and uniqueness conclusions follow.
Finally, summing \eqref{eq:channel} against $\theta_i$ gives
\eqref{eq:mixture}.
\end{proof}

The posterior under a uniform input is $a/d$; this identifies the
construction with the balanced Blackwell measure
\cite{blackwell,nasser}. Conditional reconstruction also preserves
statistical decision risks and row divergences. We may therefore optimize
over anchored laws. Allowing all of $\fM_d$ enlarges the finite-output
class, but every optimizer explicitly constructed below has finite
support.

For $\alpha>1$ and probability measures $Q,P$, write
\[
D_\alpha(Q\|P)=\frac{\log\int(dQ/dP)^\alpha\,dP}{\alpha-1}
\]
when $Q\ll P$, with value $+\infty$ otherwise. In template coordinates set
\begin{equation}\label{eq:budget}
E_{ij}(\rho)=\int a_j^\alpha a_i^{1-\alpha}\,d\rho-1,
\qquad \mathcal E_\alpha(\rho)=\max_{i\ne j}E_{ij}(\rho).
\end{equation}
Here $a=\one+Hx$. At zero coordinates use the perspective convention
$0^\alpha 0^{1-\alpha}=0$ and $b^\alpha0^{1-\alpha}=+\infty$ for $b>0$.
The budget $\mathcal E_\alpha\le C$ is exactly
\begin{equation}\label{eq:renyi-conversion}
\max_{i\ne j}D_\alpha(W_j\|W_i)\le R,
\qquad C=e^{(\alpha-1)R}-1.
\end{equation}
All averages, mixtures, and concave envelopes in this paper use the
\emph{linear} budget $C$, not the logarithmic budget $R$.
If $\mathcal E_\alpha<\infty$, all coordinates $a_i$ are positive
$\rho$-almost surely: on a zero coordinate some other coordinate is positive,
so one of the finitely many ordered-pair integrals would otherwise be infinite.
At $\alpha=2$, $E_{ij}=\chi^2(W_j\|W_i)$, where
$\chi^2(Q\|P)=\int(dQ/dP-1)^2\,dP$.

\subsection{Two sampling models and one calibrated estimator}
Let
\[
\Sigma=\int xx^\top\rho(dx),\qquad
B_i=\int(1+\gamma_i^\top x)xx^\top\rho(dx).
\]
When $\Sigma$ is invertible, define the canonical estimator from messages
$X_1,\ldots,X_n$ by
\begin{equation}\label{eq:canonical}
\widehat\theta_\rho=\one/d+H\Sigma^{-1}\bar X,
\qquad \bar X=n^{-1}\sum_{r=1}^nX_r.
\end{equation}
It is calibrated to the chosen channel. It need not take values in
$\Delta_d$. Its exact risk is the canonical objective, rather than the
minimum risk over all estimators. If $\Sigma$ is singular we assign this
canonical objective the value $+\infty$; this convention says nothing
about the risk of other estimators.

In the i.i.d. model, inputs have common distribution $\theta$ and messages
have law $q_\theta^{\otimes n}$. In the fixed-composition model the inputs
are deterministic, with counts $n\theta_i$. Write
$\Delta_{d,n}=\{\theta\in\Delta_d:n\theta_i\in\mathbb Z\}$ for the
realizable compositions. Risks are denoted $R^{\mathrm{iid}}_\theta$ and
$R^{\mathrm{fc}}_\theta$.

\begin{proposition}[Exact canonical risks]\label{prop:risks}
For invertible $\Sigma$, the estimator \eqref{eq:canonical} is unbiased
in both models and
\begin{align}
R^{\mathrm{fc}}_\theta(\rho)
&=\frac1n\sum_i\theta_i
\left[\tr(\Sigma^{-1}B_i\Sigma^{-1})-\frac{d-1}{d}\right],
\quad\theta\in\Delta_{d,n},\label{eq:fc-risk}\\
R^{\mathrm{iid}}_\theta(\rho)
&=\frac1n\int(1+h_\theta^\top x)x^\top\Sigma^{-2}x\,\rho(dx)
-\frac{\|h_\theta\|^2}{n}.
\label{eq:iid-risk}
\end{align}
The fixed-composition worst case is a vertex, and
\begin{equation}\label{eq:fc-average}
\frac1d\sum_iR^{\mathrm{fc}}_{e_i}(\rho)
=\frac1n\left[\tr\Sigma^{-1}-\frac{d-1}{d}\right].
\end{equation}
\end{proposition}

\begin{proof}
Under row $i$, $\E_iX=\Sigma\gamma_i$ and
$\Cov_iX=B_i-\Sigma\gamma_i\gamma_i^\top\Sigma$.
Thus $\E\bar X=\Sigma h_\theta$ in either model, proving
unbiasedness. For fixed composition,
$\Cov\bar X=n^{-1}\sum_i\theta_i\Cov_iX$. Taking the trace after
left and right multiplication by $\Sigma^{-1}$ proves
\eqref{eq:fc-risk}, since $\|\gamma_i\|^2=(d-1)/d$.
Its affine extension to $\Delta_d$ is maximized at a vertex; all vertices
belong to $\Delta_{d,n}$. Averaging over $i$ and using
$d^{-1}\sum_iB_i=\Sigma$ gives \eqref{eq:fc-average}.
For i.i.d. sampling, the second moment is
$M_\theta=\int(1+h_\theta^\top x)xx^\top\rho(dx)$ and
$\Cov\bar X=n^{-1}(M_\theta-\Sigma h_\theta h_\theta^\top\Sigma)$.
The same trace calculation gives \eqref{eq:iid-risk}.
\end{proof}

The covariance $\Sigma$ is also the one-message Fisher information at
the uniform input. In the local parameterization
$\theta(u)=\one/d+Hu$, the score is $x/(1+u^\top x)$, which equals
$x$ at $u=0$. The $n$-message information there is $n\Sigma$.

\section{Global symmetrization}\label{sec:symmetrization}
For a permutation matrix $\Pi$, write $P_\pi=H^\top\Pi H$ and define
\[
\bar\rho=\frac1{d!}\sum_{\pi\in S_d}(P_\pi)_\#\rho.
\]
The measure $\rho$ is exchangeable if $\rho=\bar\rho$.
The action preserves $\cX_d$ because
$1+\gamma_i^\top P_\pi x=1+\gamma_{\pi^{-1}i}^\top x$.

\begin{theorem}[Symmetrization of the worst canonical risk]\label{thm:symmetrization}
Let $d\ge3$, $\alpha>1$, and $n\ge1$. For every $\rho\in\fM_d$,
\begin{align}
\mathcal E_\alpha(\bar\rho)&\le\mathcal E_\alpha(\rho),\label{eq:sym-budget}\\
\Sigma(\bar\rho)&=\frac{\tr\Sigma(\rho)}{d-1}I_{d-1},\label{eq:sym-cov}\\
\sup_{\theta\in\Delta_d}R^{\mathrm{iid}}_\theta(\bar\rho)
&\le\sup_{\theta\in\Delta_d}R^{\mathrm{iid}}_\theta(\rho),\label{eq:sym-iid}\\
\sup_{\theta\in\Delta_{d,n}}R^{\mathrm{fc}}_\theta(\bar\rho)
&\le\sup_{\theta\in\Delta_{d,n}}R^{\mathrm{fc}}_\theta(\rho).
\label{eq:sym-fc}
\end{align}
If $\rho$ is exchangeable and $\tr\Sigma>0$, then
\begin{align}
R^{\mathrm{iid}}_\theta(\rho)
&=\frac{(d-1)^2}{n\tr\Sigma}-\frac{\|h_\theta\|^2}{n},\label{eq:exch-iid}\\
R^{\mathrm{fc}}_\theta(\rho)
&=\frac1n\left[\frac{(d-1)^2}{\tr\Sigma}-\frac{d-1}{d}\right].
\label{eq:exch-fc}
\end{align}
\end{theorem}

\begin{proof}
Let $c_{ij}(x)=a_j(x)^\alpha a_i(x)^{1-\alpha}$.
Then $c_{ij}(P_\pi x)=c_{\pi^{-1}i,\pi^{-1}j}(x)$.
Each directed budget of $\bar\rho$ is an average of directed budgets
of $\rho$, proving \eqref{eq:sym-budget}, also in the extended-real case.

The matrix $H\Sigma(\bar\rho)H^\top$ is invariant under all coordinate
permutations and annihilates $\one$. It therefore has the form
$a(I_d-\one\one^\top/d)$. Its trace is $\tr\Sigma(\rho)$, so
$a=\tr\Sigma(\rho)/(d-1)$, proving \eqref{eq:sym-cov}.
For an exchangeable law and a permutation-invariant scalar function $g$,
the vector $\int xg(x)\rho(dx)$ is fixed by every $P_\pi$ and hence is
zero. Apply this to $g(x)=\|x\|^2$. Since $\Sigma$ is scalar, the cubic
term in \eqref{eq:iid-risk} vanishes. The same cancellation in
\eqref{eq:fc-risk} proves \eqref{eq:exch-iid}--\eqref{eq:exch-fc}.

If $\Sigma(\rho)$ is singular the two risk inequalities are immediate
from the canonical convention. Otherwise, for its eigenvalues
$s_1,\ldots,s_{d-1}>0$,
\[
\frac{(d-1)^2}{\tr\Sigma(\rho)}\le\sum_{j=1}^{d-1}s_j^{-1}
=\tr\Sigma(\rho)^{-1}.
\]
The left side is $n$ times the worst i.i.d. risk of $\bar\rho$;
the right side is $nR^{\mathrm{iid}}_{\one/d}(\rho)$.
This proves \eqref{eq:sym-iid}. Subtracting $(d-1)/d$ and using
\eqref{eq:fc-average}, which is at most the maximum vertex risk,
proves \eqref{eq:sym-fc}. The zero-trace case is covered by the same
singular convention.
\end{proof}

The theorem symmetrizes the \emph{worst-case} canonical objective.
Uniform-point Fisher information alone would not supply this conclusion
without the exact risk identities. No equivariance restriction on the
original channel is needed.

\section{The exact frontier up to order two}\label{sec:global-frontier}

\subsection{Average cost and the fixed-second-moment inequality}
For a positive template $a\in\cK_d$ write
\[
P_q(a)=\sum_i a_i^q,\qquad S(a)=P_2(a)-d,\qquad
c_\alpha(a)=\frac{P_\alpha(a)P_{1-\alpha}(a)-d^2}{d(d-1)}.
\]
Set $c_\alpha=+\infty$ on the boundary of $\cK_d$.

\begin{lemma}[Trace and average linear cost]\label{lem:template}
Let $\nu$ be the law of $a=\one+Hx$ under an anchored law $\rho$. Then
\begin{equation}\label{eq:template-moments}
\tr\Sigma=\int S\,d\nu,\qquad
\frac1{d(d-1)}\sum_{i\ne j}E_{ij}(\rho)=\int c_\alpha\,d\nu.
\end{equation}
Every directed cost equals the average for an exchangeable law.
\end{lemma}
\begin{proof}
The trace identity is \eqref{eq:norm-bound}. For positive $a$,
$\sum_{i\ne j}a_j^\alpha a_i^{1-\alpha}
=P_\alpha P_{1-\alpha}-\sum_i a_i
=P_\alpha P_{1-\alpha}-d$.
Subtracting $d(d-1)$ gives the second identity. The boundary convention
agrees with an infinite ordered-pair moment. The pointwise function
$c_\alpha$ is nonnegative: it is the common directed cost of the
uniform permutation orbit of $a$, hence a nonnegative divergence moment.
Thus all integrations are justified also in the extended-real case.
Permutation invariance makes the directed costs equal.
\end{proof}

\begin{lemma}[Strict singleton reduction]\label{lem:fixed-B}
Let $d\ge3$, $1<\alpha\le2$, and $d<B<d^2$.
Among positive templates with $\sum_i a_i=d$ and $\sum_i a_i^2=B$,
the product $P_\alpha P_{1-\alpha}$ is minimized precisely at the
permutations of
\begin{equation}\label{eq:singleton-t}
(h,l,\ldots,l),\qquad
h=1+(d-1)t,\quad l=1-t,\quad
t=\sqrt{\frac{B-d}{d(d-1)}}\in(0,1).
\end{equation}
Consequently a singleton preserves $S$ and weakly reduces $c_\alpha$,
with strict reduction unless the nonconstant template is already a
singleton. At $B=d$ the only template is $\one$.
\end{lemma}
\begin{proof}
The displayed template is feasible and positive. The closure of the
moment set is compact. Jensen gives $P_\alpha\ge d$, while
$P_{1-\alpha}\to\infty$ if any coordinate tends to zero.
A minimizing sequence with bounded objective therefore stays away from
the boundary and has a positive minimizing limit.

At a minimizer the constraint gradients $\one$ and $2a$ are independent,
because $B>d$. Put $\beta=1-\alpha$, $U=P_\alpha(a)>0$, and
$V=P_\beta(a)>0$. For some multipliers $\mu,\eta$,
\[
f(a_i)=\mu+2\eta a_i,\qquad
f(z)=\beta U z^{\beta-1}+\alpha Vz^{\alpha-1}.
\]
With $U,V$ fixed at that minimizer,
\begin{equation}\label{eq:f-second}
f''(z)=-\alpha(\alpha-1)(\alpha+1)U z^{-\alpha-2}
 +\alpha(\alpha-1)(\alpha-2)V z^{\alpha-3}<0.
\end{equation}
Thus $f(z)-\mu-2\eta z$ is strictly concave and has at most two
positive zeros: three zeros would contradict strict concavity at the
middle zero. The nonconstant minimizer has exactly two levels $h>l$.
Their high-level multiplicity is an integer $s\in\{1,\ldots,d-1\}$;
this list covers every nonconstant two-level template, including cases
where the high level is the majority.

Suppose $s\ge2$, and choose two high coordinates $i,j$.
The vector $v=e_i-e_j$ is tangent to both constraints. Moreover,
$DP_\alpha[v]=DP_\beta[v]=0$.
For $J=P_\alpha P_\beta$ the full second derivative includes the
cross term $2DP_\alpha[v]DP_\beta[v]$, which therefore vanishes.
For the Lagrangian
$L=J-\mu(\sum_i a_i-d)-\eta(\sum_i a_i^2-B)$, this gives
\[
D^2L[v,v]=2f'(h)-4\eta
=2\left[f'(h)-\frac{f(h)-f(l)}{h-l}\right]<0.
\]
The last strict inequality follows from \eqref{eq:f-second}.
It contradicts the necessary second-order condition on the tangent
space of the regular constraint manifold. Hence $s=1$.
The two moment constraints then determine exactly \eqref{eq:singleton-t}.
All its permutations have the same objective, proving the complete
equality characterization.
\end{proof}

\paragraph{Connection with entropy extremizers.}
For clarity, the whole statement of \cref{lem:fixed-B} can also be
deduced from \cite[Theorem~1]{berry-sanders} together with the
strictness in its proof (their Lemma~5), including uniqueness. Normalize $p=a/d$, and let $p^*$ be the
singleton probability vector of the same second moment. Choose
$0<\delta<\min_i\{p_i,p_i^*\}$ and put
$b=1-d\delta>0$, $z_i=(p_i-\delta)/b$, and $z_i^*=(p_i^*-\delta)/b$.
The vectors $z,z^*$ are probabilities, their second moments agree,
and $z^*$ still has one high coordinate and an equal remainder.
For $q<0$ or $1<q<2$, the function
$f_q(z)=(\delta+bz)^q-\delta^q$ has $f_q(0)=0$, $f_q''>0$, and
$f_q^{(3)}<0$ on $[0,1]$. With $g(z)=z^2$, its derivative as a function
of $g'=2z$ is strictly concave. The cited theorem therefore minimizes $\sum_i f_q(z_i)$ at
permutations of $z^*$, and the strictness in its proof shows that no
other vector attains the minimum.
It follows that both $\sum_i p_i^\alpha$ and
$\sum_i p_i^{1-\alpha}$ are minimized by $p^*$ when $1<\alpha<2$.
At $\alpha=2$ the positive-order factor is fixed, while the
negative-order comparison remains strict.
Both original power sums are positive, so their product has the same
unique minimizing form. Scaling by $d$ recovers \cref{lem:fixed-B}.
The proof above retains this reduction within the manuscript without
an external proof dependency.

\subsection{The singleton curve and its envelope}
For the rest of this section fix $1<\alpha\le2$ and set
$D=d(d-1)$. In the parameterization \eqref{eq:singleton-t} define
\begin{equation}\label{eq:curve-t}
C(t)=\frac{h^\alpha l^{1-\alpha}+l^\alpha h^{1-\alpha}-h-l}{d},
\qquad S(t)=Dt^2,\qquad 0\le t<1.
\end{equation}
The uniform singleton orbit has common directed cost $C(t)$.
In the ratio parameter $\lambda=h/l$ the same expressions are
\begin{equation}\label{eq:singleton-lambda}
C_1(\lambda)=\frac{H_\alpha(\lambda)}{d+\lambda-1},\qquad
S_1(\lambda)=\frac{d(d-1)(\lambda-1)^2}{(d+\lambda-1)^2},
\quad
H_\alpha(\lambda)=(\lambda^\alpha-1)(1-\lambda^{1-\alpha}).
\end{equation}

\begin{lemma}[Unique tangency and upper-branch concavity]\label{lem:curve}
The map $C(t)$ is strictly increasing from $[0,1)$ onto $[0,\infty)$.
The function $A(t)=C(t)/t^2$, extended continuously at zero, is strictly
convex, has $A'(0)<0$, and tends to infinity at one. Its unique minimizer
$t_*\in(0,1)$ is the unique solution of
\begin{equation}\label{eq:threshold}
t_* C'(t_*)=2C(t_*).
\end{equation}
Set
\begin{equation}\label{eq:constants}
C_*=C(t_*),\qquad K_{\alpha,d}=\frac{Dt_*^2}{C_*},
\qquad g(C)=Dt(C)^2,
\end{equation}
where $t(C)$ is the inverse of $C(t)$.
The function $g$ is strictly concave for $C\ge C_*$, with tangent
$K_{\alpha,d}C$ at $C_*$. Its smallest concave majorant on
$[0,\infty)$ is the strictly increasing function
\begin{equation}\label{eq:frontier}
F_{\alpha,d}(C)=
\begin{cases}
K_{\alpha,d}C,&0\le C\le C_*,\\
d(d-1)t(C)^2,&C\ge C_*.
\end{cases}
\end{equation}
For $C>0$, $g(C)\le K_{\alpha,d}C$ with equality only at $C=C_*$.
\end{lemma}
\begin{proof}
Put $k=\alpha(\alpha-1)>0$, $p=2-\alpha\ge0$, and $q=\alpha+1$.
For affine $h,l$, direct differentiation gives
\[
(h^\alpha l^{1-\alpha})''
=k h^{\alpha-2}l^{-\alpha-1}(h'l-hl')^2.
\]
Here $h'l-hl'=d$, and therefore
\begin{equation}\label{eq:curve-derivatives}
C''(t)=kd\{h^{-p}l^{-q}+l^{-p}h^{-q}\}>0.
\end{equation}
Both summands have strictly positive second derivative. Indeed, for
$u_1=h^{-p}l^{-q}$ and $u_2=l^{-p}h^{-q}$,
\[
(\log u_1)''=\frac{p(d-1)^2}{h^2}+\frac q{l^2}>0,\qquad
(\log u_2)''=\frac p{l^2}+\frac{q(d-1)^2}{h^2}>0,
\]
and $u_j''=u_j\{[(\log u_j)']^2+(\log u_j)''\}>0$.
Thus $C^{(4)}(t)>0$. At zero,
\[
C(0)=C'(0)=0,\quad C''(0)=2kd,\quad C'''(0)=-3kd(d-2).
\]
Since $C''>0$ and $C(t)\to\infty$ as $t\uparrow1$, the asserted
strictly increasing bijection follows.

Taylor's integral formula yields
\[
A(t)=\int_0^1(1-u)C''(ut)\,du,\qquad
A''(t)=\int_0^1u^2(1-u)C^{(4)}(ut)\,du>0.
\]
Differentiation is justified on every compact subinterval of $[0,1)$,
including at zero, because $h,l$ are bounded away from zero there.
In particular $A(0)=kd$ and $A'(0)=-kd(d-2)/2<0$.
Also $A(t)\to\infty$ as $t\uparrow1$, so strict convexity gives a
unique interior minimum, characterized by \eqref{eq:threshold}.

For positive $C$, inverse differentiation gives
\[
g'(C)=\frac{2Dt}{C'(t)},\qquad
g''(C)=\frac{2D\{C'(t)-tC''(t)\}}{[C'(t)]^3}.
\]
Writing $C(t)=t^2A(t)$ shows
$tC''-C'=t^2(3A'+tA'')>0$ for $t\ge t_*$.
Hence $g$ is strictly concave above the threshold, and
$g'(C_*)=g(C_*)/C_*=K_{\alpha,d}$.
The unique minimum of $A$ also gives $g(C)/C\le K_{\alpha,d}$,
strictly except at $C_*$. These facts prove concavity and monotonicity
of \eqref{eq:frontier}. Any concave majorant of $g$ lies above the chord
from $(0,0)$ to $(C_*,g(C_*))$ and above $g$ on the upper branch.
Thus \eqref{eq:frontier} is the smallest such majorant.
\end{proof}

The ratio threshold is equivalently the unique $\lambda_*>1$ satisfying
\[
(\lambda_*-1)(\lambda_*+d-1)H_\alpha'(\lambda_*)
=(\lambda_*+2d-1)H_\alpha(\lambda_*).
\]
No closed form for noninteger $\alpha$ is needed.

\subsection{The global bound and every equality case}
For $p\in[0,1]$ and $t\in(0,1)$ define the ARR law
\begin{equation}\label{eq:augmented}
\rho_{p,t}=(1-p)\delta_0+\frac pd\sum_{i=1}^d\delta_{dt\gamma_i}.
\end{equation}
Its nonnull templates are the permutations of \eqref{eq:singleton-t}.
Explicitly, on outputs $\{\bot,1,\ldots,d\}$ its channel is
\[
W(\bot\mid i)=1-p,\qquad
W(j\mid i)=p\left(\frac{1-t}{d}+t\,\one_{\{i=j\}}\right).
\]
For $p=1$ this is GRR, with ratio $\lambda=h/l$.
Directly,
\begin{equation}\label{eq:aug-values}
\Sigma(\rho_{p,t})=pd t^2I_{d-1},\quad
\tr\Sigma=pDt^2,\quad E_{ij}(\rho_{p,t})=pC(t)\quad(i\ne j).
\end{equation}

\begin{theorem}[The full frontier; Theorem B]\label{thm:global-frontier}
For fixed $d\ge3$, $1<\alpha\le2$, and $0\le C<\infty$,
\[
\sup_{\rho\in\fM_d:\,\mathcal E_\alpha(\rho)\le C}\tr\Sigma(\rho)
=F_{\alpha,d}(C).
\]
The same maximum holds over finite-output channels.
At $C=0$ the unique anchored law is $\delta_0$. For $C>0$ the unique
maximizing anchored law is \eqref{eq:augmented}, with
\[
(p,t)=
\begin{cases}
(C/C_*,t_*),&0<C\le C_*,\\
(1,t(C)),&C\ge C_*.
\end{cases}
\]
Every directed cost equals $C$. Uniqueness is among all anchored laws;
channels may differ only by statistically identical common output
refinements.
\end{theorem}
\begin{proof}
Finite budget makes every template positive almost surely.
The strict fixed-second-moment comparison implies
\begin{equation}\label{eq:pointwise-frontier}
S(a)\le g(c_\alpha(a))\le F_{\alpha,d}(c_\alpha(a)).
\end{equation}
To see the left inequality, choose the singleton with the same signal;
its cost is at most $c_\alpha(a)$, and $g$ is strictly increasing.
Equality requires a singleton, except for $a=\one$.
By \cref{lem:template}, Jensen's inequality, and monotonicity,
\begin{equation}\label{eq:global-chain}
\tr\Sigma
\le\int F_{\alpha,d}(c_\alpha)\,d\nu
\le F_{\alpha,d}\!\left(\int c_\alpha\,d\nu\right)
\le F_{\alpha,d}(\mathcal E_\alpha(\rho))
\le F_{\alpha,d}(C).
\end{equation}
The cost is integrable and nonnegative; $F_{\alpha,d}$ is bounded by
$d(d-1)$, so the inequalities are well-defined.
The laws in the statement attain the bound by \eqref{eq:aug-values}.

For completeness we separate the two equality restrictions below the
threshold. Strict increase in \eqref{eq:global-chain} gives
$\int c_\alpha\,d\nu=\mathcal E_\alpha(\rho)=C$.
If $0<C\le C_*$, the supporting line of $F_{\alpha,d}$ is
$K_{\alpha,d}c$. It agrees with $F_{\alpha,d}$ precisely on
$[0,C_*]$: strict upper-branch concavity makes the inequality strict
for $c>C_*$. Equality in Jensen therefore gives
$c_\alpha(a)\le C_*$ almost surely. Next, on this interval equality
in the pointwise bound is $S(a)=K_{\alpha,d}c_\alpha(a)$.
The unique positive contact of the singleton curve with this line
is $C_*$ by \cref{lem:curve}; \cref{lem:fixed-B} supplies its strict
template equality case. Thus costs in $(0,C_*)$ are excluded.
Only the null template and the singleton templates at $t_*$ remain.
Their total informative mass must be $p=C/C_*$.

If $C>C_*$, the tangent to $F_{\alpha,d}$ at $C$ touches it only
there. Strict concavity gives this on the upper branch. On the lower
branch the tangent minus $K_{\alpha,d}c$ is minimized at $C_*$,
where it is strictly positive. Equality in Jensen therefore forces
$c_\alpha=C$ almost surely. Equality in \eqref{eq:pointwise-frontier}
then forces a singleton at $t(C)$.

In either case let $w_i$ be the mass of templates whose high coordinate
is $i$, with total mass $p$. Anchoring says
\[
(1-p)+pl+(h-l)w_i=1.
\]
Since $h+(d-1)l=d$, this forces $w_i=p/d$ for every $i$.
The anchored law is therefore exactly the one stated.
All its directed costs equal $C$, proving budget saturation as well.
At $C=0$, every directed divergence is zero, so all rows coincide;
their normalized templates equal $\one$ almost surely and $\rho=\delta_0$.
\end{proof}

\begin{corollary}[The chi-squared boundary case]\label{cor:chi2}
At $\alpha=2$ the threshold ratio and slope are
\begin{equation}\label{eq:chi2-constants}
\lambda_*=\sqrt{d-1},\qquad
K_d=\frac{d(d-1)}{d+2\sqrt{d-1}},\qquad
C_*=\left(1-\frac1{\sqrt{d-1}}\right)^2.
\end{equation}
For $C\ge C_*$ let $\lambda(C)>1$ solve
\[
\frac{(\lambda-1)^2(\lambda+1)}
{\lambda(d+\lambda-1)}=C.
\]
Then the upper frontier equals
$d(d-1)(\lambda(C)-1)^2/(d+\lambda(C)-1)^2$.
The solution is unique; equivalently it is the root above one of
\[
\lambda^3-(1+C)\lambda^2-[1+(d-1)C]\lambda+1=0.
\]
\end{corollary}
\begin{proof}
At order two, $H_2(\lambda)=(\lambda-1)^2(\lambda+1)/\lambda$.
Thus
$S_1/C_1=d(d-1)\lambda/[(\lambda+1)(d+\lambda-1)]$.
Its derivative has the sign of $d-1-\lambda^2$.
Substitution gives \eqref{eq:chi2-constants}.
The uniqueness of $\lambda(C)$ follows from the strict increase of
$C(t)$ and of $\lambda(t)$; clearing denominators gives the cubic.
\end{proof}

\section{Failure of the singleton family above order two}\label{sec:above-two}

For any $\alpha>1$, an orbit with $s$ high coordinates, ratio
$\lambda>1$, and low level $z=d/[d+s(\lambda-1)]$ has
\begin{align}
S_s(\lambda)&=\frac{d\,s(d-s)(\lambda-1)^2}
 {[d+s(\lambda-1)]^2},\label{eq:subset-signal}\\
E_s(\lambda)&=\frac{s(d-s)H_\alpha(\lambda)}
 {(d-1)[d+s(\lambda-1)]}.\label{eq:subset-cost}
\end{align}
Indeed,
$P_\alpha P_{1-\alpha}-d^2
=zs(d-s)(\lambda^\alpha+\lambda^{1-\alpha}-\lambda-1)
=zs(d-s)H_\alpha(\lambda)$.
The orbit makes \eqref{eq:subset-cost} the cost of every ordered pair.
Here $s=1,\ldots,d-1$ covers every nonconstant two-level template,
always counting the high coordinates. No reduction to two levels is
asserted for arbitrary minimizers when $\alpha>2$.

\begin{theorem}[Explicit failure for every higher order; Theorem C]
\label{thm:above-two}
Let $\alpha>2$ and define
\[
b_\alpha=(2^\alpha-1)(1-2^{1-\alpha}),\qquad
T_\alpha=\frac92(4b_\alpha^2)^{1/(\alpha-2)},\qquad
s=\max\{2,\lfloor T_\alpha\rfloor+1\},\quad d=2s.
\]
Set $C=s b_\alpha/[3(2s-1)]$.
The finite channel on all $s$-subsets $A\subset[2s]$ given by
\begin{equation}\label{eq:explicit-subset}
W(A\mid i)=\frac{(4/3)\one_{\{i\in A\}}+(2/3)\one_{\{i\notin A\}}}
 {\binom{2s}{s}}
\end{equation}
has common directed cost $C$ and trace $2s/9$.
Every anchored law supported on singleton templates and the null
template, with maximum directed cost at most $C$, has trace at most
$(dC)^{2/\alpha}$, and
\[
\frac{2s}{9}>(dC)^{2/\alpha}.
\]
Thus \eqref{eq:explicit-subset} beats the entire singleton envelope,
including all mixtures and visible null thinning, at R\'enyi budget
$R=\log(1+C)/(\alpha-1)$.
\end{theorem}
\begin{proof}
Put $q=2/\alpha\in(0,1)$. If $\lambda\ge1$ and $u=1/\lambda$,
\[
H_\alpha(\lambda)
=\lambda^\alpha(1-u^\alpha)(1-u^{\alpha-1})
\ge\lambda^\alpha(1-u)^2
\ge(\lambda-1)^\alpha.
\]
For a singleton write $L=d+\lambda-1$. From
\eqref{eq:singleton-lambda},
\[
\frac{S_1(\lambda)^{\alpha/2}}{dC_1(\lambda)}
=\frac{[d(d-1)]^{\alpha/2}(\lambda-1)^\alpha}
 {dL^{\alpha-1}H_\alpha(\lambda)}
\le\frac{d^{\alpha-1}}{L^{\alpha-1}}\le1.
\]
At $\lambda=1$ the resulting inequality is $0\le0$.
It follows pointwise that $S_1\le(dC_1)^q$.
The value $C_1(\lambda)$ is the average template cost independently
of the position of the high coordinate. Thus for any anchored
singleton-supported law, not necessarily exchangeable, Jensen gives
\[
\int S\,d\nu\le d^q\int c_\alpha^q\,d\nu
\le\left(d\int c_\alpha\,d\nu\right)^q
\le(dC)^q.
\]
This includes arbitrary parameter mixtures and null mass, with no
assumption about the shape of their concave envelope.

Half of all $s$-subsets of $[2s]$ contain a fixed $i$, so
\eqref{eq:explicit-subset} has row sums one. Its templates have $s$
coordinates $4/3$ and $s$ coordinates $2/3$.
Substitution into \eqref{eq:subset-signal}--\eqref{eq:subset-cost},
with $\lambda=2$, gives $S=2s/9$ and the stated common cost.
For $s\ge2$,
\[
dC=\frac{2s^2b_\alpha}{3(2s-1)}\le\frac{4sb_\alpha}{9}.
\]
Since $s>T_\alpha$, we have
$s^{\alpha-2}>4b_\alpha^2(9/2)^{\alpha-2}$, or equivalently
$(2s/9)^\alpha>(4sb_\alpha/9)^2$.
Taking $\alpha$th roots proves the strict comparison.
Every dimension, output alphabet, and budget used here is finite.
\end{proof}

The subset orbit is exchangeable, so its larger trace also strictly
reduces worst-case canonical risk in both sampling models relative to
every feasible singleton-supported law. Indeed, the uniform i.i.d. risk
and the average vertex fixed-composition risk bound those risks below
by the reciprocal-trace expressions in \eqref{eq:exch-iid}--\eqref{eq:exch-fc},
whereas the subset orbit attains them.

In particular, the negative direction uses an explicit finite
construction, rather than a limit in $d$. Its dimension bound is
deliberately coarse and is not a minimal-dimension assertion.

\begin{remark}[Numerical dimensions; no minimality assertion]\label{rem:numerics}
The ancillary script searches the subset orbits beyond the singleton
family. The table records the smallest \emph{found} dimension on the
specified search, not a proved minimum. Every accepted numerical
example is separately compared with an upper supporting line for the
entire continuous singleton family; this comparison is evaluated in
floating arithmetic and is not a formal interval certificate.
\[
\begin{array}{c|r|r|r|r}
\alpha& d&s&\lambda\text{ (approximately)}
 &\text{relative excess (approximately)}\\\hline
2.05&2628&2&36.8125&0.0000030\\
2.1 &716 &2&19.2026&0.0000057\\
2.2 &206 &2&10.2519&0.000113\\
2.5 &46  &2&4.79702&0.000873\\
3   &18  &2&2.97991&0.00304\\
4   &9   &2&2.08495&0.00873\\
8   &5   &2&1.49035&0.0321
\end{array}
\]
For $\alpha=2.01,2.02$, no violation was found through $d=10^4$.
Every integer dimension from three to the reported dimension, or to
$10^4$ for these two orders, was searched. All $s=2,\ldots,d-1$ were
included for $d\le128$; larger dimensions used $s=2,\ldots,64$ and
48 rounded geometric subset sizes. At each size the search used 282
log-ratio grid points and local optimization near the five best entries.
The relative excess compares the subset signal with the computed
singleton supporting-line upper bound at the same cost. Failed
searches establish neither absence of violations nor optimality.
\end{remark}

\begin{proposition}[An exact five-symbol separation]\label{prop:alpha8}
At $\alpha=8$, the subset orbit with $d=5$, $s=2$, and
$\lambda=3/2$ has
\[
S=\frac5{24},\qquad C=\frac{12981995}{2239488}.
\]
Every anchored law supported on singleton templates and the null
template satisfies
\[
\tr\Sigma\le\frac{\mathcal E_8(\rho)}{46}+\frac{19}{250}.
\]
Consequently the indicated subset orbit strictly exceeds the entire
singleton envelope at its budget, by at least
$81314369/12877056000$ in trace.
\end{proposition}
\begin{proof}
The subset values follow from
\eqref{eq:subset-signal}--\eqref{eq:subset-cost}.
Appendix~\ref{app:certificate} proves the affine singleton inequality
by an explicit rational polynomial certificate, with all finite
arithmetic specified. Integrating that inequality preserves its
constant term since the total singleton mass is at most one, and the
average cost is at most the maximum directed cost. Finally,
\[
\frac5{24}-\frac1{46}\frac{12981995}{2239488}-\frac{19}{250}
=\frac{81314369}{12877056000}>0.
\]
\end{proof}

\section{Exact unbiased and asymptotic minimax risk}\label{sec:estimation}

Throughout this section, $d\ge3$, $1<\alpha\le2$, and the channel is
common to all users and known to the analyst. The budget is the linear
quantity $\mathcal E_\alpha(\rho)\le C$. A R\'enyi budget $R$ is
converted by $C=\exp((\alpha-1)R)-1$ before applying any formula.
The finite-sample result for arbitrary unbiased estimators concerns the
i.i.d. model; the fixed-composition assertion below concerns the
canonical estimator only.

\subsection{The exact canonical objective}
\begin{theorem}[Canonical optimum at every finite budget]\label{thm:canonical-opt}
For $0<C<\infty$ and $n\ge1$, the canonical risks satisfy
\begin{align}
\inf_{\rho:\mathcal E_\alpha(\rho)\le C}
\sup_{\theta\in\Delta_d}R^{\mathrm{iid}}_\theta(\rho)
&=\frac{(d-1)^2}{nF_{\alpha,d}(C)},\label{eq:canonical-iid-opt}\\
\inf_{\rho:\mathcal E_\alpha(\rho)\le C}
\sup_{\theta\in\Delta_{d,n}}R^{\mathrm{fc}}_\theta(\rho)
&=\frac1n\left[\frac{(d-1)^2}{F_{\alpha,d}(C)}
                    -\frac{d-1}{d}\right].\label{eq:canonical-fc-opt}
\end{align}
Both infima are attained by the unique anchored law in
\cref{thm:global-frontier}. Thus the optimal channel is augmented
randomized response for $C\le C_*$ and ordinary randomized response
for $C\ge C_*$. Uniqueness is understood at the level of the anchored
law, allowing common statistically uninformative output refinements.
\end{theorem}
\begin{proof}
For an invertible covariance, the uniform-input i.i.d. risk and the
average vertex fixed-composition risk give, respectively,
\[
\sup_\theta R^{\mathrm{iid}}_\theta(\rho)
\ge\frac1n\tr\Sigma^{-1}
\ge\frac{(d-1)^2}{n\tr\Sigma},\qquad
\sup_\theta R^{\mathrm{fc}}_\theta(\rho)
\ge\frac1n\left[\frac{(d-1)^2}{\tr\Sigma}-\frac{d-1}{d}\right].
\]
The second assertion uses \eqref{eq:fc-average}; every vertex is in
$\Delta_{d,n}$. The frontier theorem bounds $\tr\Sigma$ by
$F_{\alpha,d}(C)$. Singular laws have infinite canonical objective by
convention. The frontier optimizer is exchangeable and attains the
bounds by \eqref{eq:exch-iid}--\eqref{eq:exch-fc}. Any optimizer must
have trace $F_{\alpha,d}(C)$, because a smaller trace would make its
risk lower bound strictly larger. The uniqueness assertion follows
from \cref{thm:global-frontier}.
\end{proof}

\subsection{All globally unbiased estimators at finite sample size}
An estimator is globally unbiased if all its components are integrable
and $\E_{\theta,\rho}\widehat\theta=\theta$ for every
$\theta\in\Delta_d$. In this subsection estimators take values in
$\R^d$; the definition does not require their values to be probabilities.

\begin{theorem}[Exact finite-sample unbiased minimax risk; Theorem D]\label{thm:unbiased-minimax}
For every $n\ge1$ and $0<C<\infty$, in the i.i.d. model,
\begin{equation}\label{eq:unbiased-minimax}
\inf_{\rho:\mathcal E_\alpha(\rho)\le C}
\inf_{\substack{\widehat\theta\in\R^d:\\
 \E_{\theta,\rho}\widehat\theta=\theta\ \forall\theta\in\Delta_d}}
\sup_{\theta\in\Delta_d}
\E_{\theta,\rho}\|\widehat\theta-\theta\|^2
=\frac{(d-1)^2}{nF_{\alpha,d}(C)}.
\end{equation}
The canonical estimator on the channel in
\cref{thm:global-frontier} attains this value. Its anchored law is the
unique channel law participating in an optimum. For any fixed channel
with invertible $\Sigma$, every globally unbiased estimator satisfies
\begin{equation}\label{eq:fixed-channel-unbiased}
\E_{\one/d,\rho}\|\widehat\theta-\one/d\|^2
\ge\frac{\tr\Sigma^{-1}}n
\ge\frac{(d-1)^2}{n\tr\Sigma}.
\end{equation}
If $\Sigma$ is singular, no globally unbiased integrable estimator
exists.
\end{theorem}
\begin{proof}
Put $m=d-1$ and $\theta(u)=\one/d+Hu$. At $u=0$, messages have law
$\rho$. If an estimator has infinite risk at $\one/d$, the lower
bounds are immediate. Otherwise set
$\widehat u=H^\top(\widehat\theta-\one/d)$, which is square
integrable under $\rho^{\otimes n}$. The sample density at $u$ is
\[
f_u(x_1,\ldots,x_n)=\prod_{r=1}^n(1+u^\top x_r).
\]
On a sufficiently small cube around zero this polynomial and all its
derivatives of order one are uniformly bounded in $x_1,\ldots,x_n$, by
\eqref{eq:norm-bound}. Its derivative times $\widehat u$ is therefore
dominated by a constant times the integrable variable
$\|\widehat u\|$. Differentiating
$\E_u\widehat u=u$ at zero gives the exact score identity
\begin{equation}\label{eq:unbiased-score}
\E_0[\widehat u Z^\top]=I_m,\qquad
Z=\sum_{r=1}^nX_r,\qquad \E_0[ZZ^\top]=n\Sigma.
\end{equation}
If $\Sigma$ is singular, there is a nonzero $v$ with
$v^\top X=0$ $\rho$-almost surely. For every sufficiently small $t$,
$q_{tv}=q_0$ although $\theta(tv)\ne\theta(0)$. This contradicts
global unbiasedness of an integrable estimator. This argument does
not require finite variance.

For invertible $\Sigma$, put
$V=\widehat u-(n\Sigma)^{-1}Z$. Equation
\eqref{eq:unbiased-score} yields $\E_0[VZ^\top]=0$, and hence
\[
\E_0[\widehat u\widehat u^\top]
=(n\Sigma)^{-1}+\E_0[VV^\top]\succeq(n\Sigma)^{-1}.
\]
Orthogonal projection onto $T_d$ cannot increase the squared norm.
Taking traces proves the left inequality in
\eqref{eq:fixed-channel-unbiased}; applying
$\sum_{j=1}^ms_j^{-1}\ge m^2/\sum_{j=1}^ms_j$ to the positive
eigenvalues of $\Sigma$ proves the second. The frontier theorem now
gives the lower bound in \eqref{eq:unbiased-minimax}. The canonical
estimator is globally unbiased and attains it by
\cref{thm:canonical-opt}. Finally, a channel participating in an
optimum must have $\tr\Sigma=F_{\alpha,d}(C)$, since otherwise
\eqref{eq:fixed-channel-unbiased} would be strictly larger than the
optimum. The equality case in \cref{thm:global-frontier} gives its
unique anchored law.
\end{proof}

\begin{remark}[Why unbiased estimates may leave the simplex]\label{rem:no-simplex-unbiased}
At finite budget all row laws are mutually absolutely continuous.
Suppose a globally unbiased estimate were $\Delta_d$-valued. At
$\theta=e_i$, its nonnegative coordinates and expectation $e_i$
would force it to equal $e_i$ almost surely. Equivalence of the
$n$-message laws at $e_i$ and $e_j$ would force the same value at
$e_j$, contradicting unbiasedness for $j\ne i$. Thus no such estimate
exists. Projection onto $\Delta_d$ may reduce risk, but does not
preserve global unbiasedness.
\end{remark}

\subsection{A uniform nonasymptotic lower bound}
The next statement applies to all estimators, including biased ones.
Write
\[
\mathfrak R_{n,\alpha,d}(C)
=\inf_{\rho:\mathcal E_\alpha(\rho)\le C}\inf_{\widehat\theta}
  \sup_{\theta\in\Delta_d}
  \E_{\theta,\rho}\|\widehat\theta-\theta\|^2,
\]
where the inner infimum is over all $\R^d$-valued estimators.

\begin{proposition}[Explicit van Trees bound]\label{prop:finite-minimax-bound}
Let $m=d-1$, $R_d=\sqrt{d(d-1)}$, and $0<C<\infty$. For every
$n\ge1$ and every $0<r<(R_d\sqrt m)^{-1}$,
\begin{equation}\label{eq:finite-bayes}
\mathfrak R_{n,\alpha,d}(C)
\ge\frac{m^2}{\displaystyle
 nF_{\alpha,d}(C)/(1-R_d\sqrt m\,r)+10m/r^2}.
\end{equation}
The same bound holds if the channel infimum is restricted to finite
output alphabets or if the estimator infimum is restricted to
$\Delta_d$-valued estimates.
\end{proposition}
\begin{proof}
Use the local parameterization $\theta(u)=\one/d+Hu$. The cube
$[-r,r]^m$ lies in the interior of the parameter simplex: indeed,
$\|\gamma_i\|=\sqrt{m/d}$, so
$|\gamma_i^\top u|\le mr/\sqrt d<1/d$.
Put a product prior on the cube, with one-coordinate density
$r^{-1}\varphi(u_j/r)$, where
\begin{equation}\label{eq:prior-density}
\varphi(t)=\frac{15}{16}(1-t^2)^2\one_{[-1,1]}(t),\qquad
\int_{-1}^1\frac{\varphi'(t)^2}{\varphi(t)}\,dt=10.
\end{equation}
The normalization follows by integrating $1-2t^2+t^4$;
$\varphi'^2/\varphi=15t^2$ in $(-1,1)$ gives the information
constant. Thus the prior information is $10m/r^2$.

For a feasible anchored law, the $n$-message Fisher information is
\[
I_n(u,\rho)=n\int\frac{xx^\top}{1+u^\top x}\,\rho(dx).
\]
The denominator is at least $1-\kappa$ uniformly, with
$\kappa=R_d\sqrt m\,r<1$. All parameter derivatives on the cube
are bounded by integrable constants independent of $\rho$;
this justifies differentiation under the integral, also when
$\Sigma$ is singular. By \cref{thm:global-frontier},
\[
\tr I_n(u,\rho)\le\frac{n\tr\Sigma}{1-\kappa}
\le\frac{nF_{\alpha,d}(C)}{1-\kappa}.
\]
For any estimator, let
$\widehat u=H^\top(\widehat\theta-\one/d)$. Then
$\|\widehat\theta-\theta(u)\|^2\ge\|\widehat u-u\|^2$.
Apply \cref{lem:van-trees}. Its integration by parts is valid for
this polynomial prior, which vanishes at the boundary, and the bounded
model scores; for infinite Bayes risk the bound is immediate.
The Bayes risk is at most the supremum risk. The resulting bound is
uniform in $\rho$, so both infima preserve it. Restricting either
infimum can only increase its value.
\end{proof}

\begin{remark}[Block averaging and the fixed-channel bound]\label{rem:block-averaging}
The block argument used by Park, Nam, and Lee
\cite{park-nam-lee} gives another route to the second inequality in
\eqref{eq:fixed-channel-unbiased}. Fix $n$, a channel with
$T=\tr\Sigma>0$, and a globally unbiased estimator with finite
worst-case risk $M$. Average its estimates from $b$ independent blocks
of size $n$. Cross terms vanish by unbiasedness, so the resulting
estimator on $bn$ messages has worst-case risk exactly $M/b$.
The proof of \cref{prop:finite-minimax-bound}, with its actual trace
$T$ in place of $F_{\alpha,d}(C)$ and
$r=(bnT)^{-1/4}$, gives
\[
\frac Mb\ge\frac{m^2}{bnT}(1-o_b(1)).
\]
Multiplying by $b$ before letting $b\to\infty$, with $n$ and the
channel fixed, yields $M\ge m^2/(nT)$. There is no interchange of an
infimum and a limit. The direct score proof above is stronger, giving
$\tr\Sigma^{-1}/n$ without an asymptotic step.
\end{remark}

\subsection{All estimators and arbitrary budget sequences}
For i.i.d. inputs the unordered sample is sufficient for
$q_\theta^{\otimes n}$: conditional on its empirical measure, a
uniform random ordering is independent of $\theta$. Under squared
error, conditional averaging therefore transfers any ordered-sample
estimator to a no-worse unordered-sample estimator. The following
result is consequently unchanged by shuffling the messages.

\begin{theorem}[Sharp minimax asymptotics; Theorem E]\label{thm:minimax}
Fix $d\ge3$ and $1<\alpha\le2$. For any sequence $C_n>0$ such that
$nF_{\alpha,d}(C_n)\to\infty$,
\begin{equation}\label{eq:minimax}
\mathfrak R_{n,\alpha,d}(C_n)
=\frac{(d-1)^2}{nF_{\alpha,d}(C_n)}(1+o(1)).
\end{equation}
The same asymptotic value holds with finite-output channels, with
$\Delta_d$-valued estimators, or with both restrictions. The canonical
estimator on the channel in \cref{thm:global-frontier} supplies the
upper bound; its Euclidean projection supplies a simplex-valued upper
bound.
\end{theorem}
\begin{proof}
The exact canonical risk in \cref{thm:canonical-opt} is an upper
bound, attained with at most $d+1$ outputs. Euclidean projection onto
$\Delta_d$ cannot increase loss at any $\theta\in\Delta_d$.
For the lower bound put $T_n=F_{\alpha,d}(C_n)$ and
$r_n=(nT_n)^{-1/4}$. For all sufficiently large $n$ this is an
admissible radius in \cref{prop:finite-minimax-bound}. Since
$\kappa_n=R_d\sqrt{d-1}\,r_n\to0$ and
$10(d-1)/r_n^2=o(nT_n)$, that proposition gives
\[
\mathfrak R_{n,\alpha,d}(C_n)
\ge\frac{(d-1)^2}{nT_n/(1-\kappa_n)+10(d-1)/r_n^2}
=\frac{(d-1)^2}{nT_n}(1-o(1)).
\]
The error is uniform over all feasible anchored laws. Thus it remains
valid for finite-output channels, and restricting the estimator to
the simplex cannot weaken the lower bound.
\end{proof}

\begin{corollary}[Low-budget specialization]\label{cor:low-budget}
For fixed $d\ge3$ and $1<\alpha\le2$, if $C_n\to0$ and
$nC_n\to\infty$, then
\[
\mathfrak R_{n,\alpha,d}(C_n)
=\frac{(d-1)^2}{nK_{\alpha,d}C_n}(1+o(1)).
\]
For all sufficiently large $n$, an attaining canonical channel has
activation probability $C_n/C_*$ and singleton parameter $t_*$.
For $\alpha=2$ the leading term equals
$(d-1)(d+2\sqrt{d-1})/(ndC_n)$.
\end{corollary}
\begin{proof}
Eventually $C_n<C_*$, so the assertion follows from
$F_{\alpha,d}(C_n)=K_{\alpha,d}C_n$ and \cref{thm:minimax}.
\end{proof}

These theorems keep $d$ and $\alpha$ fixed. They do not assert an exact
finite-sample optimum over biased estimators, nor uniform asymptotics
when $d$ or $\alpha$ varies with $n$.

\section{Scope and further questions}\label{sec:discussion}

The channel constraint and the statistical model are separate choices.
The R\'enyi budget is a maximum over local row pairs, whereas the loss is
a worst-case estimation risk. The trace and canonical results, the exact
finite-sample theorem for $\R^d$-valued unbiased estimators, and the
all-estimators asymptotic theorem have different scopes.
The comparison in Appendix~\ref{raw:appendix} concerns a pure local
likelihood-ratio cap. Appendix~\ref{app:contamination} concerns a restricted
homogeneous neighboring-data experiment. It does not turn the local
budget into a general shuffled-privacy guarantee.

The information-growth condition in \cref{thm:minimax} is material.
For example, with $d=3$, consider the anchored laws
\[
 \rho_n=(1-1/n)\delta_0+\frac1{3n}\sum_{i=1}^3
                   \delta_{(3/4)\gamma_i}.
\]
These are augmented randomized response with ratio $2$ and activation
$1/n$. At the two fixed interior mixture coordinates $u=0$ and
$v=\gamma_1/2$, the one-message likelihood ratio $dq_v/dq_0$ is $1$,
$5/4$, or $7/8$. Under $q_0^{\otimes n}$, the counts of the three active
symbols converge to independent
$N_1,N_2,N_3\sim\operatorname{Pois}(1/3)$. Hence
\[
 \log\frac{dq_v^{\otimes n}}{dq_0^{\otimes n}}
 \ \Longrightarrow\
 N_1\log(5/4)+(N_2+N_3)\log(7/8),
\]
a nondegenerate compound-Poisson law. Indeed its characteristic function
is the limit of
\[
 \left[1+\frac{(5/4)^{it}/3+2(7/8)^{it}/3-1}{n}\right]^n.
\]
For every fixed finite order $\alpha>1$, the maximum linearized
pairwise cost of $\rho_n$ is a constant times $1/n$. For
$1<\alpha\le2$, therefore, $nF_{\alpha,d}(C_n)$ remains bounded in this
example. It does not contradict \cref{thm:minimax}. Uniformly bounded
likelihood ratios alone do not exclude non-Gaussian triangular-array
limits.

The optimal trace for $\alpha>2$ is not determined here. The subset
construction proves failure of the entire singleton envelope at suitable
alphabet sizes and budgets, without classifying the optimizers in that
range. Other questions left open are exact finite-sample minimax risk
over all estimators, including biased estimators, and regimes with
$d=d_n\to\infty$. Interactive and multi-message protocols are outside
the common independent-channel model considered here.

\section*{Use of AI tools}
The author formulated the problem, directed the research, and decided
which results to include. ChatGPT (OpenAI) assisted with proof strategies,
drafting and revising the text, and verification scripts. Claude
(Anthropic) assisted with red-team audits of several versions, including
the final one, the research report, counterexamples to the privacy
envelope, and independent verification scripts. These systems were used
as assistants, not as authors. The author checked all results and takes
full responsibility for the content of this paper.

\appendix
\section{A coordinatewise van Trees inequality}\label{app:van-trees}
\begin{lemma}\label{lem:van-trees}
Let $U\in\R^m$ have a continuously differentiable prior density $\pi$
on a rectangle, positive in its interior and zero on its boundary.
Let $p_u(y)$ be a dominated observation model with Fisher information
$I(u)$. Assume differentiation under the integral and integration by
parts are justified, and the model and prior scores are square
integrable. For every estimator $\widehat U$,
\[
\int\E_u\|\widehat U-u\|^2\pi(u)\,du
\ge\frac{m^2}{\int\tr I(u)\pi(u)\,du+J(\pi)},\qquad
J(\pi)=\int\frac{\|\nabla\pi(u)\|^2}{\pi(u)}\,du.
\]
\end{lemma}
\begin{proof}
This standard form of the van Trees bound \cite{van-trees} has a short
coordinatewise proof. Under the joint density $\pi(u)p_u(y)$, write
$S_j=\partial_j\log p_U(Y)$ and $Q_j=\partial_j\log\pi(U)$.
Integration by parts in $u_j$, whose boundary term vanishes, gives
\[
\E[(\widehat U_j-U_j)(S_j+Q_j)]=1.
\]
The model score has conditional mean zero, so its cross term with the
prior score vanishes. Consequently
\[
\E(S_j+Q_j)^2=\int I_{jj}(u)\pi(u)\,du
                  +\int\frac{(\partial_j\pi(u))^2}{\pi(u)}\,du.
\]
Cauchy--Schwarz bounds the coordinate risk below by the reciprocal of
this expression. Summing the $m$ bounds and applying
$\sum_jb_j^{-1}\ge m^2/\sum_jb_j$ proves the claim. For infinite
risk the result is immediate; otherwise the stated square-integrability
assumptions justify the products.
\end{proof}

\section{An exact rational singleton separator}\label{app:certificate}
For $\alpha=8$ and $d=5$, the singleton signal and cost are
\[
 G(z)=\frac{20(z-1)^2}{(z+4)^2},\qquad
 E(z)=\frac{z^8-z-1+z^{-7}}{z+4},\qquad z\geq1.
\]
We claim that $G(z)\leq E(z)/46+19/250$ for every $z\geq1$.
Multiplication of the difference by the positive quantity
$11500z^7(z+4)^2$ gives
\begin{align*}
 P(z)&=874z^7(z+4)^2-230000z^7(z-1)^2\\
 &\quad+250(z+4)(z^{15}-z^8-z^7+1)\\
 &=1000+250z-217016z^7+465742z^8-229376z^9
       +1000z^{15}+250z^{16}.
\end{align*}
The coefficients of $P(2+x)$, in ascending powers of $x$, are
\begin{gather*}
23164892,\ 228046586,\ 983490304,\ 2490515328,\ 4195222528,\\
5039489056,\ 4507592208,\ 3085884712,\ 1643696974,\ 686170624,\\
224224000,\ 56784000,\ 10920000,\ 1540000,\ 150000,\ 9000,\ 250.
\end{gather*}
They are positive, proving $P(z)>0$ for $z\geq2$.
For completeness the finite rational check on $[1,2]$ is specified as
follows. Write $P(z)=\sum_{k=0}^{16}p_kz^k$. On an interval $[a,b]$ put
\[
 q_j=\sum_{k=j}^{16}p_k\binom{k}{j}a^{k-j}(b-a)^j,
 \qquad
 b_i=\sum_{j=0}^i q_j\frac{\binom{i}{j}}{\binom{16}{j}}.
\]
Then
$P(a+(b-a)u)=\sum_{i=0}^{16}b_i\binom{16}{i}u^i(1-u)^{16-i}$.
Direct substitution of the seven nonzero $p_k$ listed above gives the
following exact minima of the seventeen coefficients:
\[
\begin{array}{c|c}
 [a,b]&\min_{0\leq i\leq16}b_i\\\hline
 \left[1,3/2\right]&459778997/32768\\
 \left[3/2,13/8\right]&56924840527405/153545080832\\
 \left[13/8,7/4\right]&10693822598760399325/140737488355328\\
 \left[7/4,2\right]&2379632237622701/2147483648
\end{array}
\]
All are positive. Since the Bernstein basis is nonnegative and sums to
one, this proves $P>0$ on $[1,2]$, and hence the claimed affine bound.
The bound also holds for every probability mixture of singleton channels
and the null channel: integrate the bound and use that the total
informative mass is at most one. Finally the $s=2$, $\lambda=3/2$ channel
has
\[
 S=\frac5{24},\qquad C=\frac{12981995}{2239488},\qquad
 S-\frac C{46}-\frac{19}{250}
   =\frac{81314369}{12877056000}>0.
\]
Thus it strictly exceeds the entire singleton envelope at this budget.

\section{Comparison with a pure local privacy cap}\label{raw:appendix}

This appendix records the same moment calculation under a pure local
differential privacy constraint.  Throughout, $d\ge3$ is fixed and
$\lambda=e^{\varepsilon_0}>1$.  The subset-selection optimizer and its exact
finite-sample i.i.d.\ risk are established results: Park, Nam, and Lee
\cite[Theorems~3.3 and~3.9 of the arXiv version]{park-nam-lee} give the risk formula and
optimality among all unbiased local schemes.  Ye and Barg
\cite[Theorem~III.1]{ye-barg-sharp} prove the corresponding sharp
all-estimators asymptotic constant for fixed $d$ and $\varepsilon_0$.
We restate the calculation in the present coordinates to compare the two
constraints, and then give a version in which $\varepsilon_0$ may vary with
the sample size.  The formulas below are included to compare a pointwise likelihood-ratio
cap with the moment constraints in the main text.

Write $a=\one+Hx$ for a template and
$\nu=a_\#\rho$ for its law.  Since the row densities relative to the
uniform-input output law are $a_i$, the channel is
$\varepsilon_0$-LDP if and only if $\nu$ is supported on
\[
 \cK_{d,\lambda}
 :=\{a\in\R_+^d:\ \textstyle\sum_i a_i=d,
                         \ a_i\le\lambda a_j\ \text{for all }i,j\}.
\]
Indeed, each measure inequality between two rows is equivalent to the
corresponding inequality between these densities almost everywhere.
There are finitely many row pairs, so all the inequalities hold on a
common full-measure set.  The multiplicities $s=1,\ldots,d-1$ include every possible nonconstant
two-level template: $s$ always counts the larger coordinates.
Every point of $\cK_{d,\lambda}$ has strictly
positive coordinates: a zero coordinate would force all coordinates to
vanish.  The set is a compact polytope.

\begin{lemma}[Ratio-capped vertices]\label{raw:vertices}
The vertices of $\cK_{d,\lambda}$ are exactly the permutations of
\[
 a^{(s)}_\lambda
 =\frac{d}{d+s(\lambda-1)}
       (\underbrace{\lambda,\ldots,\lambda}_{s},
        \underbrace{1,\ldots,1}_{d-s}),
 \qquad 1\le s\le d-1.
\]
\end{lemma}

\begin{proof}
For a feasible $a$, let $M=\max_i a_i$ and $m=\min_i a_i>0$.
If $M<\lambda m$, all ratio inequalities are strict, so for any
\emph{distinct} indices $p,q$ the points
$a\pm t(e_p-e_q)$ are feasible for sufficiently small $t>0$.
Thus a vertex must satisfy $M=\lambda m$.

Suppose this equality holds but $m<a_r<M$ for some $r$.  Let
$U=\{i:a_i=M\}$ and $L=\{i:a_i=m\}$.  Set
\[
 v_i=\lambda\quad(i\in U),\qquad
 v_i=1\quad(i\in L),\qquad
 v_r=-\lambda|U|-|L|,
\]
and set the remaining coordinates of $v$ to zero.  Then
$\sum_i v_i=0$, and every active inequality
$a_i=\lambda a_j$ has $i\in U$, $j\in L$, so
$v_i=\lambda v_j$.  All other inequalities have positive slack.
Consequently $a\pm tv$ are distinct feasible points for small $t>0$.
A vertex therefore has only the two levels $\lambda m,m$; their sum
fixes the stated template.

Conversely, suppose such a two-level point is a nontrivial convex
combination of feasible points $b,c$.  Each equality
$a_i=\lambda a_j$, with $i$ high and $j$ low, forces equality in both
$b_i\le\lambda b_j$ and $c_i\le\lambda c_j$.
All low coordinates of $b$ are thus equal, all high coordinates are
$\lambda$ times that value, and the sum constraint forces $b=a$.
The same applies to $c$, proving extremality.
\end{proof}

Let $\rho_{s,\lambda}$ be the uniform orbit law of the anchored point
$H^\top(a^{(s)}_\lambda-\one)$.  It has mean zero because permutation
averaging sends every template of sum $d$ to $\one$.  The corresponding
channel is the usual $s$-subset-selection mechanism.  Define
\[
 T_{d,\lambda}(s)
 :=\frac{d\,s(d-s)(\lambda-1)^2}
          {[d+s(\lambda-1)]^2},\qquad
 T_d^*(\lambda):=\max_{1\le s\le d-1}T_{d,\lambda}(s),
 \quad
 \cS_{d,\lambda}^*:=\argmax_{1\le s\le d-1}T_{d,\lambda}(s).
\]

\begin{proposition}[Raw-cap trace and canonical risk]\label{raw:canonical}
For every $n\ge1$,
\[
 \sup_{\rho:\,W_\rho\ \varepsilon_0\text{-LDP}}
             \tr\Sigma(\rho)=T_d^*(\lambda),
\]
and the optimal worst-case canonical risks are
\begin{align*}
 \inf_{\rho:\,W_\rho\ \varepsilon_0\text{-LDP}}
     \sup_{\theta\in\Delta_d}R^{\mathrm{iid}}_\theta(\rho)
   &=\frac{(d-1)^2}{nT_d^*(\lambda)},\\
 \inf_{\rho:\,W_\rho\ \varepsilon_0\text{-LDP}}
     \sup_{\theta\in\Delta_{d,n}}R^{\mathrm{fc}}_\theta(\rho)
   &=\frac1n\left[\frac{(d-1)^2}{T_d^*(\lambda)}
                         -\frac{d-1}{d}\right].
\end{align*}
For fixed composition the supremum is over compositions realizable with
$n$ records.  Every mixture of $\rho_{s,\lambda}$ with
$s\in\cS_{d,\lambda}^*$ attains these values.  These mixtures are
exactly the exchangeable trace or canonical-risk optimizers; no
classification of all nonsymmetric optimizers is asserted.
\end{proposition}

\begin{proof}
The trace is
$\int(\sum_i a_i^2-d)\,\nu(da)$.
The strictly convex function $\sum_i a_i^2-d$ attains its maximum
over the polytope only at vertices.  Substitution of the vertices from
\cref{raw:vertices} gives $T_{d,\lambda}(s)$.
Every maximizing vertex orbit is an anchored law, so this upper bound
is attained.  Equality in the integral requires support on maximizing
vertices.  Exchangeability then gives precisely the stated mixtures.

Permutation averaging preserves the ratio cap and the trace, and the
canonical-risk symmetrization inequalities in the main text permit
restriction to exchangeable laws.  For such a law of positive trace,
$\Sigma=\tr\Sigma\,I_{d-1}/(d-1)$ and the worst-case canonical risks
are $(d-1)^2/(n\tr\Sigma)$ and
$[(d-1)^2/\tr\Sigma-(d-1)/d]/n$, respectively.
Both decrease strictly with the trace, proving the formulas and the
exchangeable equality cases.  The fixed-composition maximum can be
taken over realizable compositions because all simplex vertices are
realizable for every $n$.
\end{proof}

The maximizing subset sizes can be found by comparing at most two
integers.  Indeed, the derivative of the continuous factor is
\[
 \frac{d}{ds}\frac{s(d-s)}{[d+s(\lambda-1)]^2}
 =\frac{d[d-(\lambda+1)s]}{[d+s(\lambda-1)]^3}.
\]
If $d/(\lambda+1)\le1$, the maximum is at $s=1$.
Otherwise every maximizing integer is a floor or ceiling of
$d/(\lambda+1)$ within $\{1,\ldots,d-1\}$.  Possible ties are retained.
For the subset law of size $s$, the exact i.i.d.\ risk at any input distribution
is
\[
 R^{\mathrm{iid}}_\theta(\rho_{s,\lambda})
 =\frac1n\left[
 \frac{(d-1)^2[d+s(\lambda-1)]^2}
      {d\,s(d-s)(\lambda-1)^2}
 +\frac1d-\sum_i\theta_i^2\right].
\]
This is exactly the formula in
\cite[Theorem~3.3, equations~(23)--(24), arXiv version]{park-nam-lee} after
identifying their parameters $v,k,e^\varepsilon$ with $d,s,\lambda$.

\begin{corollary}[A varying pure local cap]\label{raw:varying}
Fix $d\ge3$.  Let $\lambda_n=e^{\varepsilon_{0,n}}>1$ satisfy
$nT_d^*(\lambda_n)\to\infty$.  In the i.i.d.\ multinomial model,
\[
 \inf_{\rho:\,W_\rho\ \varepsilon_{0,n}\text{-LDP}}
 \inf_{\widehat\theta}\sup_{\theta\in\Delta_d}
 \E_{\theta,\rho}\|\widehat\theta-\theta\|_2^2
 =\frac{(d-1)^2}{nT_d^*(\lambda_n)}(1+o(1)).
\]
The inner infimum is over all estimators.  The assertion includes
fixed caps, whose sharp constant is already known
\cite{ye-barg-sharp}; no uniformity in growing $d$ is asserted.
\end{corollary}

\begin{proof}
The upper bound follows from \cref{raw:canonical}.  For completeness,
the lower bound is uniform over every channel allowed by the cap.
Write $m=d-1$, $T_n=T_d^*(\lambda_n)$, and
$r_n=(nT_n)^{-1/4}$.  For large $n$, the cube $[-r_n,r_n]^m$
parameterizes interior compositions by $\theta(u)=\one/d+Hu$.
Place on this cube the product prior obtained by scaling the density
$\varphi(t)=(15/16)(1-t^2)^2\one_{[-1,1]}(t)$.
Its prior Fisher information is $10m/r_n^2$.

For an arbitrary feasible anchored law, one observation has density
$1+u^\top x$ relative to $\rho$, and $\|x\|\le R_d:=\sqrt{d(d-1)}$.
Hence the $n$-sample Fisher information obeys, uniformly on the cube,
\[
 \tr I_u(\rho)
 =n\int\frac{\|x\|^2}{1+u^\top x}\,\rho(dx)
 \le\frac{nT_n}{1-R_d\sqrt m\,r_n}.
\]
For any estimator $\widehat\theta$, orthogonal projection to
$\widehat u=H^\top(\widehat\theta-\one/d)$ does not increase squared
error.  The coordinatewise van Trees inequality used in the main
minimax proof therefore gives
\[
 \sup_\theta\E_{\theta,\rho}
       \|\widehat\theta-\theta\|^2
 \ge
 \frac{m^2}{nT_n/(1-R_d\sqrt m\,r_n)+10m/r_n^2}
 =\frac{m^2}{nT_n}(1-o(1)).
\]
Since $d$ is fixed, $r_n\to0$ and $r_n^{-2}=o(nT_n)$.
The remainder is independent of $\rho$, so the bound survives both
infima.  The histogram is sufficient for this i.i.d.\ product family,
so the argument applies equally to anonymous observations.
\end{proof}

\section{Homogeneous contamination and its limitation}\label{app:contamination}

This appendix records the exact binary comparison that is available for a
homogeneous background. Its envelope is a consequence of classical binary
experiment simulation \cite{kov-composition,murtagh-vadhan}; it is not a bound
for all neighboring shuffled datasets.

Let $P,Q$ be two probability laws on a finite alphabet $\mathcal Y$, and let
$n\ge1$. Define the ordered, symmetrized experiments
\begin{equation}\label{eq:contam-pair}
 Q_{0,n}=P^{\otimes n},\qquad
 Q_{1,n}=\frac1n\sum_{m=1}^n
 P^{\otimes(m-1)}\otimes Q\otimes P^{\otimes(n-m)}.
\end{equation}
These are the output laws for the datasets $i^n$ and $j i^{n-1}$ after
uniformly shuffling independent messages from rows $P$ and $Q$.
Both laws are permutation-invariant, so publishing a histogram rather than
an ordered shuffled sample gives an equivalent binary experiment.
For mutually absolutely continuous laws $A,B$, use
\[
 D_f(A\|B)=\E_B f\!\left(\frac{dA}{dB}\right),\qquad
 H_\eta(A,B)=\E_B\left(\frac{dA}{dB}-\eta\right)_+,
 \quad \eta\ge1,
\]
where $f:[0,\infty)\to\R$ is convex and $f(1)=0$.

\begin{proposition}[The contamination likelihood ratio]\label{prop:contam-lr}
If $Q\ll P$, and $L=dQ/dP$ under $P$, then under $Q_{0,n}$,
\begin{equation}\label{eq:contam-average}
 \frac{dQ_{1,n}}{dQ_{0,n}}=\overline L_n
 :=\frac1n\sum_{m=1}^n L_m,
\end{equation}
where the $L_m$ are independent copies of $L$. The same identity holds for
the likelihood ratio of the histogram laws.
\end{proposition}
\begin{proof}
The likelihood ratio of the $m$th summand in \eqref{eq:contam-pair} is
$L(y_m)$. Averaging gives \eqref{eq:contam-average}, which is a function of
the histogram. Equivalently, conditional on a histogram, either ordered law
is uniform over all its compatible ordered samples.
\end{proof}

\begin{proposition}[Exact forward chi-squared dilution]\label{prop:chi2-dilution}
Suppose $Q\ll P$ and $\chi^2(Q\|P)<\infty$. For the homogeneous
pair in \eqref{eq:contam-pair}, including its histogram version,
\begin{align}\label{eq:chi2-dilution}
 \chi^2(Q_{1,n}\|Q_{0,n})
 &=\E_{Q_{0,n}}(\overline L_n-1)^2
 =\frac{\Var_P(L)}n=\frac{\chi^2(Q\|P)}n,\\
 D_2(Q_{1,n}\|Q_{0,n})
 &=\log\left(1+\frac{\chi^2(Q\|P)}n\right).
\end{align}
The reverse chi-squared divergence need not equal
$\chi^2(P\|Q)/n$.
\end{proposition}
\begin{proof}
Since $\E_P L=1$, independence in \eqref{eq:contam-average} gives
$\E(\overline L_n-1)^2=\Var_P(L)/n$; the formula for $D_2$ is its
definition. For the reverse assertion, take
\[
 P=(2/3,1/3),\qquad Q=(1/3,2/3),\qquad n=2.
\]
The histogram probabilities, indexed by the number $0,1,2$ of outputs
in the second symbol, are
\[
 Q_{0,2}=\tfrac19(4,4,1),\qquad
 Q_{1,2}=\tfrac19(2,5,2).
\]
Both local directed chi-squared divergences are $1/2$, whereas
\[
 \chi^2(Q_{1,2}\|Q_{0,2})=\frac14,
 \qquad
 \chi^2(Q_{0,2}\|Q_{1,2})
 =\frac29+\frac1{45}+\frac1{18}=\frac3{10}.
\]
Thus forward and reverse directions behave differently.
\end{proof}

For a channel, maximizing \eqref{eq:chi2-dilution} over the ordered
homogeneous row pairs gives
$\log(1+\chi_*(W)/n)$ for the maximum forward order-two R\'enyi
divergence, where $\chi_*(W)=\max_{i\ne j}\chi^2(W_j\|W_i)$.
Consequently the local budget and this shuffled divergence are related by
an explicit monotone transformation; they are not numerically equal.
This statement concerns the forward homogeneous pairs only.

Fix $\varepsilon_0\ge0$, and put $a=e^{-\varepsilon_0}$,
$b=e^{\varepsilon_0}$. The binary randomized-response rows are
\[
 B_0=\left(\frac b{b+1},\frac1{b+1}\right),\qquad
 B_1=\left(\frac1{b+1},\frac b{b+1}\right).
\]
Write $Q_{0,n}^{\mathrm{BRR}},Q_{1,n}^{\mathrm{BRR}}$ for
\eqref{eq:contam-pair} with $P=B_0,Q=B_1$. For $b>1$, the law of
$dB_1/dB_0$ under $B_0$ is
\begin{equation}\label{eq:contam-endpoints}
 \mu^\star=\frac{b-1}{b-a}\,\delta_a+
             \frac{1-a}{b-a}\,\delta_b.
\end{equation}
For $b=1$, set $\mu^\star=\delta_1$.

\begin{theorem}[Binary envelope for homogeneous contamination]
\label{thm:contam-envelope}
Suppose $P,Q$ satisfy the two-sided local cap
$a\le dQ/dP\le b$ almost surely under $P$. Then, for every $f$ as above,
\[
 D_f(Q_{1,n}\|Q_{0,n})
 \le D_f(Q_{1,n}^{\mathrm{BRR}}\|Q_{0,n}^{\mathrm{BRR}}),
\]
and the analogous inequality holds with both arguments exchanged.
In particular, both directed $H_\eta$ profiles are bounded by the
corresponding BRR profiles for every $\eta\ge1$. All these bounds are
attained by $P=B_0,Q=B_1$.
\end{theorem}
\begin{proof}
If $b=1$, then $P=Q$ and every divergence is zero. For $b>1$, define
probability laws
\[
 K_0=\frac{bP-Q}{b-1},\qquad K_1=\frac{bQ-P}{b-1}.
\]
The local cap ensures nonnegativity, and each numerator has mass $b-1$.
Moreover,
\[
 P=\frac{bK_0+K_1}{b+1},\qquad
 Q=\frac{K_0+bK_1}{b+1}.
\]
Thus a single Markov kernel, with output laws $K_0,K_1$ on the two binary
inputs, maps $B_0,B_1$ to $P,Q$. This is the pure-privacy case of the
classical simulation principle; see \cite[Lemma~3.2]{murtagh-vadhan} and
\cite{kov-composition}.
Applying this kernel independently to each coordinate maps the BRR
contamination pair to \eqref{eq:contam-pair}. The kernel is identical at
every coordinate and therefore also defines a postprocessing of the BRR
histogram. For completeness, under the base law the output likelihood ratio is
the conditional expectation of the input likelihood ratio given the
output of this kernel. Conditional Jensen therefore gives the forward
$f$-divergence inequality; applying the same argument with the two base
laws exchanged gives the reverse inequality.
For $H_\eta$, take $f(t)=(t-\eta)_+$, which has $f(1)=0$ when
$\eta\ge1$, and exchange the two laws for the reverse profile.
\end{proof}

\begin{proposition}[Equality of both complete profiles]\label{prop:contam-rigidity}
Under the hypotheses of \cref{thm:contam-envelope}, suppose, for one fixed
$n\ge1$, that
\[
 H_\eta(Q_{1,n},Q_{0,n})
 =H_\eta(Q_{1,n}^{\mathrm{BRR}},Q_{0,n}^{\mathrm{BRR}}),
 \qquad \eta\ge1,
\]
and that the same equalities hold after exchanging both pairs of arguments.
Then the law of $L=dQ/dP$ under $P$ is exactly $\mu^\star$.
If $\varepsilon_0>0$, merging outputs with the same value of $L$ yields
the two-row BRR experiment; the original pair is recovered by a
row-independent conditional refinement.
\end{proposition}
\begin{proof}
The case $b=1$ is immediate. Otherwise let $Z=\overline L_n$, and define
$Z^\star$ using independent variables with law \eqref{eq:contam-endpoints}.
Both variables lie in $[a,b]$ and have mean one. By
\cref{prop:contam-lr}, the forward profiles are their upper stop-loss
transforms $S_Z(t)=\E(Z-t)_+$ for $t\ge1$. The reverse profiles give
their lower transforms, since
\[
 H_\eta(Q_{0,n},Q_{1,n})
 =\E(1-\eta Z)_+
 =\eta\,\E(\eta^{-1}-Z)_+.
\]
Consequently the lower transforms $T_Z(t)=\E(t-Z)_+$ agree for
$0<t\le1$. The identity $S_Z(t)-T_Z(t)=1-t$ recovers equality of
$S_Z,S_{Z^\star}$ also for $0<t\le1$; for $t\le0$ both equal $1-t$.
Thus their full upper stop-loss transforms agree. These transforms
determine the law, because their right derivatives are
$-\PP(Z>t)$. In particular, $\Var(Z)=\Var(Z^\star)$.
Independence in \eqref{eq:contam-average} now gives
\[
 \Var(L)=n\Var(Z)=(1-a)(b-1).
\]
Since $a\le L\le b$ and $\E L=1$,
\[
 \E[(L-a)(b-L)]=(1-a)(b-1)-\Var(L)=0.
\]
The integrand is nonnegative, so $L\in\{a,b\}$ almost surely; its
mean fixes the weights in \eqref{eq:contam-endpoints}.
On either likelihood-ratio class the restrictions of $P$ and $Q$ are
proportional, so their conditional output laws coincide. This proves
the refinement assertion.
\end{proof}

\begin{proposition}[A heterogeneous background can exceed this envelope]
\label{ex:heterogeneous}
The BRR contamination bound in \cref{thm:contam-envelope} is not a
universal bound for arbitrary neighboring datasets under the same local cap.
\end{proposition}
\begin{proof}
Take $d=n=3$ and the $\log 2$-LDP channel
\[
 W=\frac14
 \begin{pmatrix}2&1&1\\1&2&1\\1&1&2\end{pmatrix}.
\]
Let $P_{\rm het},Q_{\rm het}$ be the laws of its shuffled histogram on
the neighboring datasets $(1,1,3)$ and $(2,1,3)$. Direct expansion of
the product probabilities gives the following exact values; a histogram
$h=(h_1,h_2,h_3)$ records the counts of the three output symbols.
\[
\begin{array}{c|rr@{\qquad}c|rr}
 h&64P_{\rm het}(h)&64Q_{\rm het}(h)&h&64P_{\rm het}(h)&64Q_{\rm het}(h)\\ \hline
 (0,0,3)&2&2 &(1,1,1)&14&16\\
 (0,1,2)&5&7 &(1,2,0)&5&7\\
 (0,2,1)&4&7 &(2,0,1)&12&7\\
 (0,3,0)&1&2 &(2,1,0)&8&7\\
 (1,0,2)&9&7 &(3,0,0)&4&2
\end{array}
\]
Hence
\[
 H_1(P_{\rm het},Q_{\rm het})
 =\frac12\sum_h|P_{\rm het}(h)-Q_{\rm het}(h)|=\frac5{32}.
\]
For BRR with $b=2$ and $n=3$, the histogram laws of the number of
outputs equal to symbol $1$, indexed by $0,1,2,3$, are
\[
 Q_{0,3}^{\mathrm{BRR}}:\ \frac1{27}(1,6,12,8),\qquad
 Q_{1,3}^{\mathrm{BRR}}:\ \frac1{27}(2,9,12,4).
\]
Their total variation is $4/27$. Since $5/32-4/27=7/864>0$, the
homogeneous envelope fails for this heterogeneous neighboring pair.
\end{proof}

The distinction is essential: the independent-average identity
\eqref{eq:contam-average} uses an identical row at every position under
the base experiment. The comparison proved here therefore supplies no
worst-case shuffle-DP guarantee over arbitrary common backgrounds.


\begin{thebibliography}{99}

\bibitem{mironov}
I.~Mironov.
\newblock R\'enyi differential privacy.
\newblock In \emph{2017 IEEE 30th Computer Security Foundations Symposium (CSF)},
  pages 263--275, 2017.
\newblock doi:10.1109/CSF.2017.11; \texttt{arXiv:1702.07476v3}.

\bibitem{duchi-ruan}
J.~C. Duchi and F.~Ruan.
\newblock The right complexity measure in locally private estimation:
  it is not the Fisher information.
\newblock \texttt{arXiv:1806.05756v3}, 2020.

\bibitem{gilani-lrdp}
A.~Gilani, G.~R. Kurri, O.~Kosut, and L.~Sankar.
\newblock Unifying privacy measures via maximal $(\alpha,\beta)$-leakage
  (M$\alpha$beL).
\newblock \texttt{arXiv:2304.07456v2}, 2024.

\bibitem{wang-xu-regression}
D.~Wang and J.~Xu.
\newblock On sparse linear regression in the local differential privacy model.
\newblock In \emph{Proceedings of the 36th International Conference on Machine
  Learning}, volume 97 of \emph{Proceedings of Machine Learning Research},
  pages 6628--6637, 2019. Supplementary material, Definitions 3--4.

\bibitem{blackwell}
D.~Blackwell.
\newblock Equivalent comparisons of experiments.
\newblock \emph{Annals of Mathematical Statistics}, 24(2):265--272, 1953.
\newblock doi:10.1214/aoms/1177729032.

\bibitem{nasser}
R.~Nasser.
\newblock Continuity of channel parameters and operations under various DMC topologies.
\newblock \texttt{arXiv:1701.04466v2}, 2017.

\bibitem{nam-park-lee-geometry}
S.-H. Nam, H.-Y. Park, and S.-H. Lee.
\newblock Optimal privacy-utility trade-offs in LDP:
functional and geometric perspectives.
\newblock \texttt{arXiv:2605.02319v1}, 2026.

\bibitem{berry-sanders}
D.~W. Berry and B.~C. Sanders.
\newblock Bounds on general entropy measures.
\newblock \emph{Journal of Physics A: Mathematical and General},
36(49):12255--12265, 2003.
\newblock doi:10.1088/0305-4470/36/49/008;
\texttt{arXiv:quant-ph/0305059v3}.

\bibitem{sakai-iwata}
Y.~Sakai and K.-i. Iwata.
\newblock Sharp bounds between two R\'enyi entropies of distinct positive orders.
\newblock Preprint, \texttt{arXiv:1605.00019v2}, 2016.

\bibitem{shvets-growing}
A.~Shvets.
\newblock Growing alphabets in canonical shuffle experiments:
  likelihood-ratio laws, estimation bounds, and low-budget equivariant design.
\newblock \texttt{arXiv:2603.18080v2}, 2026.

\bibitem{shvets-fixed}
A.~Shvets.
\newblock Fixed-composition shuffle asymptotics in the full-support Gaussian regime.
\newblock \texttt{arXiv:2602.09029v6}, 2026.

\bibitem{barnes-chen-ozgur}
L.~P. Barnes, W.-N. Chen, and A.~\"Ozg\"ur.
\newblock Fisher information under local differential privacy.
\newblock \texttt{arXiv:2005.10783v1}, 2020.

\bibitem{kov-extremal}
P.~Kairouz, S.~Oh, and P.~Viswanath.
\newblock Extremal mechanisms for local differential privacy.
\newblock \emph{Journal of Machine Learning Research}, 17(17):1--51, 2016.
\newblock \texttt{arXiv:1407.1338v3}.

\bibitem{ye-barg-subset}
M.~Ye and A.~Barg.
\newblock Optimal schemes for discrete distribution estimation under locally
  differential privacy.
\newblock \emph{IEEE Transactions on Information Theory}, 64(8):5662--5676, 2018.
\newblock doi:10.1109/TIT.2018.2809790; \texttt{arXiv:1702.00610}.

\bibitem{ye-barg-sharp}
M.~Ye and A.~Barg.
\newblock Asymptotically optimal private estimation under mean square loss.
\newblock \texttt{arXiv:1708.00059v1}, 2017.

\bibitem{park-nam-lee}
H.-Y. Park, S.-H. Nam, and S.-H. Lee.
\newblock Exactly optimal and communication-efficient private estimation via
  block designs.
\newblock \emph{IEEE Journal on Selected Areas in Information Theory}, 5:123--134, 2024.
\newblock doi:10.1109/JSAIT.2024.3381195.
\newblock Extended version: \texttt{arXiv:2305.01261v3}, 2023;
  numbered results cited here refer to that arXiv version.

\bibitem{kov-composition}
P.~Kairouz, S.~Oh, and P.~Viswanath.
\newblock The composition theorem for differential privacy.
\newblock In \emph{Proceedings of the 32nd International Conference on Machine Learning},
  volume 37 of \emph{Proceedings of Machine Learning Research},
  pages 1376--1385, 2015.

\bibitem{murtagh-vadhan}
J.~Murtagh and S.~Vadhan.
\newblock The complexity of computing the optimal composition of differential privacy.
\newblock \emph{Theory of Computing}, 14(8):1--35, 2018.
\newblock doi:10.4086/toc.2018.v014a008.

\bibitem{van-trees}
H.~L. Van Trees.
\newblock \emph{Detection, Estimation, and Modulation Theory, Part~I}.
\newblock John Wiley \& Sons, 1968.

\end{thebibliography}
\end{document}